\documentclass[sigconf,nonacm]{acmart}
\usepackage[utf8]{inputenc}
\usepackage{mathrsfs}
\usepackage{amsmath}
\usepackage{stmaryrd}
\usepackage{listings}
\usepackage{multicol}
\usepackage{hyperref}
\usepackage{tabularx}
\usepackage{multirow}
\usepackage{diagbox}
\usepackage{bbold}
\usepackage{graphicx} 
\usepackage{xcolor}
\usepackage{sansmath}
\usepackage{booktabs}
\usepackage{subcaption}
\usepackage{nicefrac}

\newtheorem{theo}{Theorem}[section]

\newtheorem{rem}{Remark}[subsection]

\newcommand{\FMR}{\ensuremath{{\tt FMR}}}

\newcommand{\FMRe}{\ensuremath{{\tt \overline{FMR}}}}

\newcommand{\CIL}[1]{{\tt CIL}_{#1,\alpha}(\FMRe)}
\newcommand{\CIU}[1]{{\tt CIU}_{#1,\alpha}(\FMRe)}

\newcommand{\Ncomb}{\ensuremath{\mathscr{N}}}
\newcommand{\Kcomb}{\ensuremath{\mathscr{K}}}

\setcopyright{none}

\begin{document}
%
\title{Accuracy Limits as a Barrier to Biometric System Security}





\author{Axel DURBET}
\email{axel.durbet@uca.fr}
\affiliation{%
  \institution{UCA, LIMOS (UMR 6158 CNRS), Clermont-Ferrand,}
  \country{France}
}

\author{Paul-Marie GROLLEMUND}
\email{paul_marie.grollemund@uca.fr}
\affiliation{%
    \institution{UCA, LMBP (UMR 6620 CNRS), Clermont-Ferrand,}
  \country{France}}

\author{Pascal LAFOURCADE}
\email{pascal.lafourcade@uca.fr}
\affiliation{%
    \institution{UCA, LIMOS (UMR 6158 CNRS), Clermont-Ferrand,}
  \country{France}
}

\author{Kevin THIRY-ATIGHEHCHI}
\email{kevin.atighehchi@uca.fr}
\affiliation{%
   \institution{UCA, LIMOS (UMR 6158 CNRS), Clermont-Ferrand,}
 \country{France}}

%
%
%

\begin{abstract}
Biometric systems are widely used for identity verification and identification, including \textit{authentication} (\textit{i.e.}, one-to-one matching to verify a claimed identity) and \textit{identification} (\textit{i.e.}, one-to-many matching to find a subject in a database).  
The matching process relies on measuring similarities or dissimilarities between a fresh biometric template and enrolled templates.  
The \textit{False Match Rate} (\(\FMR\)) is a key metric for assessing the accuracy and reliability of such systems.
This work analyzes biometric systems based on their \(\FMR\), with two main contributions. First, we explore \textit{untargeted attacks}, where an adversary aims to impersonate any user within a database. We determine the number of trials required for an attacker to successfully impersonate a user and derive the \textit{critical population size} (\textit{i.e.}, the maximum number of users in the database) required to maintain a given level of security. Furthermore, we compute the \textit{critical \(\FMR\)} value required to ensure resistance against \textit{untargeted attacks} as the database size increases.
Second, we revisit the \textit{biometric birthday problem} to evaluate the approximate and exact probabilities that two users in a database collide (\textit{i.e.}, can impersonate each other). Based on this analysis, we derive both the approximate \textit{critical population size} and the \textit{critical \(\FMR\)} value required to bound the likelihood of such collisions occurring with a given probability. These thresholds offer actionable insights for designing systems that mitigate the risk of impersonation and collisions, particularly in large-scale biometric databases.
Our findings indicate that current biometric systems fail to deliver sufficient accuracy to achieve an adequate security level against \textit{untargeted attacks}, even in small-scale databases. Moreover, state-of-the-art systems face significant challenges in addressing the \textit{biometric birthday problem}, especially as database sizes grow.
\end{abstract}
\keywords{Biometric Security; Minimal Leakage Model; Birthday Problem 
}

\maketitle

\vspace{-2mm}

\section{Introduction}

Biometric systems, operating in either authentication or identification modes, leverage unique physical or behavioral characteristics to confirm or establish an individual’s identity within a database. Authentication compares a fresh biometric data with an enrolled biometric template, such as when unlocking a smartphone or logging into an account. Identification, on the other hand, entails searching for a match between an individual's biometric data and a database of enrolled records. This process is commonly employed in forensic investigations or large-scale governmental programs, such as national identification systems and border control, where rapid and accurate identification is essential. Biometric technologies are now widely adopted in everyday devices and critical applications, such as online financial platforms, passport verification, and national identification systems (\textit{e.g.}, India’s Aadhaar~\cite{rao2019aadhaar}, Estonia’s e-Residency~\cite{SULLIVAN2017470}, or Malaysia’s MyKad~\cite{loo2009user}).

Biometric templates are structured representations of raw biometric data (\textit{feature vectors}), designed to enable efficient matching and to facilitate integration 
with security technologies, including Biometric Template Protection (BTP) schemes~\cite{sandhya2017biometric,nandakumar2015biometric}.
Despite their convenience, biometric data present unique security challenges, as they are inherently not revocable once compromised, unlike passwords. 
This lack of revocability makes them highly sensitive, and their sensitivity is further amplified by vulnerabilities inherent to biometric systems.
For instance, \textit{untargeted attacks} exploit the probabilistic nature of biometric matching to compromise any template in the database without targeting a specific user, as opposed to \textit{targeted attacks}, where the adversary’s goal is to impersonate a particular individual by matching their specific template. 
In this work, we adopt the \textit{minimal information leakage} model, which ensures that the comparison process reveals only the binary outcome (\textit{i.e.}, acceptance or rejection), completely preventing adversaries from exploiting intermediate computations or states.
Such attacks achieve two objectives: ($i$)~impersonating any user in the database, effectively granting unauthorized access, and ($ii$)~compromising the template of any legitimate user by generating an attacker’s input that closely collides with the stored template. 
These collision inputs expose template details, which could be sensitive if the template is not revocable, thereby violating data protection regulations like the General Data Protection Regulation (GDPR).
In the identification mode, \textit{untargeted attacks} exploit the absence of limits on the number of attempts, unlike the authentication mode, which often enforces restrictions similar to password-based mechanisms.
Analyzing these attacks under worst-case assumptions—such as unlimited attempts—helps uncover systemic vulnerabilities and informs security improvements. 
This approach provides valuable insights into the weaknesses of biometric systems, even when practical constraints, such as enforced attempt limits, are in place.

A key metric for assessing the security of biometric systems is the False Match Rate (\FMR), which quantifies the likelihood of an incorrect acceptance or identification. 
We demonstrate how the security level depends on parameters such as the \FMR, which directly influences the probability of successful attacks, particularly as the density of enrolled templates in the database increases. 
This parameter is especially useful for determining the system's vulnerability to \textit{untargeted attacks} that leverage large database sizes to increase the likelihood of success. 
Moreover, the \FMR\ is also of interest for understanding the \textit{biometric birthday problem}, which highlights the heightened risk of template collisions in large databases, analogous to the classic \textit{birthday problem}, where the likelihood of shared birthdays increases rapidly with group size.
Our analysis of these issues and their implications on attack complexity provides useful findings for mitigating vulnerabilities. 
For instance, by quantifying the impact of database size and \FMR, we align our recommendations with principles outlined in international standards.

We draw an analogy with cryptographic hash functions to support this constructive approach. 
Cryptographic hash functions are evaluated based on parameters such as the output length, which directly influence resistance to preimage and collision attacks~\cite{NewDirectionsCrypto,rogaway2004cryptographic}. 
Similarly, the \FMR\ serves as a key parameter for quantifying attack complexities and deriving system-level recommendations. 
Just as collision resistance in hash functions informs design guidelines, the relationship between \FMR, database size, and security level can guide the construction of robust biometric systems. 
However, as with cryptographic hash functions, correctly dimensioned parameters alone do not guarantee security. 
Design flaws or insufficient entropy in biometric templates can still compromise system robustness. 
The guidelines proposed here are therefore intended to complement broader efforts in building secure and reliable systems.

Although biometric templates are stored in a form that preserves their secrecy and data comparisons are performed securely within biometric systems, their compromise remains possible under the \textit{minimal leakage} model, in both \textit{online} and \textit{offline attack} scenarios, which are relevant to the attacks studied in this paper.
In the \textit{online attack scenario}, adversaries interact with the remote biometric system to test inputs of their choice and identify matches. These attacks rely on the ability to make repeated queries and, when constrained in number, may resemble \textit{password spray attacks}~\cite{wiefling2022pump}, where an attacker minimizes the number of guesses per user while maximizing the number of users targeted.

In the \textit{offline attack scenario}, adversaries gain access to encrypted or transformed templates after a biometric database breach. Such attacks are particularly effective when the transformation relies solely on biometric data as the secret, which is the case for cryptographic obfuscations of evasive functions~\cite{galbraith2019obfuscated,galbraith2021obfuscating} and biometric cryptosystems (BC~\cite{1299169,SHARMA2023103458}). BC schemes, a subcategory of BTP schemes, often use primitives like secure sketches, fuzzy vaults, fuzzy commitments, or fuzzy extractors to either protect an existing cryptographic key or derive one directly from biometric data. These primitives can validate the correctness of reconstructed data or keys using a cryptographic hash function: the hash of the original key or biometric data is stored during enrollment and later compared to the hash of the reconstructed key or data. If the hashes match, the reconstruction is validated. However, possessing both the helping data and the hash enables offline attacks, where adversaries can brute-force or reconstruct biometric data or keys without system interaction. While the search space is limited by the entropy of the biometric data, offline attacks remove network latency, allowing unlimited, high-speed guesses. In some cases, the secrecy of non-revocable templates depends on an independent secret key, whose compromise can lead to the complete compromise of the database.

Biometric recognition systems follow standards like ISO/IEC 24745~\cite{ISO24745} and ISO/IEC 30136~\cite{ISO30136}, which define four security criteria: \emph{irreversibility}, \emph{unlinkability}, \emph{revocability}, and \emph{performance preservation}. Secure biometric systems achieve these goals through a chain of treatments. Below, we contextualize our work in relation to these principles:
\begin{itemize}
\item    \textbf{Irreversibility and revocability in protected templates:} Secure biometric templates must ensure \textit{irreversibility}, preventing reconstruction of the original biometric data from the (protected) template, and \textit{revocability}, enabling replacement of compromised templates. Revocability requires randomizing features using a replaceable and secret user-specific token (\textit{e.g.}, stored secret or password). However, not all irreversible methods allow revocability\footnote{For example, a hashed social security number ensures irreversibility (up to the entropy of the number) but is not revocable because the input is fixed. The SSN is particularly sensitive as it serves as a unique personal identifier and is very difficult to change, even in cases of compromise. Since the hash is completely determined by this sensitive input, it cannot be replaced if compromised. Adding salt prevents attackers from reusing precomputed tables (e.g., rainbow tables) but does not eliminate vulnerability to exhaustive search attacks, as each salt is public and each salted hash can still be brute-forced independently.}, even in secure biometric systems—where data are stored in a form that preserves their secrecy, and comparisons are performed securely. The sequence of treatments may include cancelable biometrics (CB~\cite{patel2015cancelable,manisha2020cancelable}), a subcategory of BTP schemes, that map biometric features into a \textit{new metric space} using a secret-based one-way transformation applied client-side.
This process enables comparisons of transformed features instead of raw data, improving entropy, supporting template revocation, and mitigating risks such as impersonation and reconstruction of raw features.
Non-revocable templates (\textit{i.e.}, those not randomized client-side) are not uniformly distributed, making them more susceptible to exhaustive search attacks. In this work, we demonstrate that untargeted attacks can be performed efficiently when the sequence of treatments excludes such a revocable transformation. This underscores the importance of client-side randomization.

    \item \textbf{Unlinkability across systems:}  
This criterion ensures that biometric templates from different systems cannot be correlated, preserving user privacy. 
However, reversing part of the templates through untargeted attacks may act as an initial step toward such linkability.


    \item \textbf{Efficient performance preservation:}  
Security measures should preserve the biometric system’s \textit{accuracy}, \textit{speed}, and \textit{scalability}.     
Our findings characterize this trade-off, showing how larger databases increase collision risks, thereby reducing the \textit{accuracy} in the \textit{identification mode}.

\end{itemize}

\paragraph{Contributions}
This paper investigates the security of biometric systems, focusing on two critical aspects: untargeted attacks and the \emph{biometric birthday problem}. By analyzing the False Match Rate (\FMR) and its confidence interval---a range of values that likely contains the true \FMR\ with a specified probability---we derive indicators to support the design of secure biometric systems. 
These indicators are: ($i$) \emph{the critical population}, which represents the largest number of entries a biometric system can handle while satisfying a predefined security level, and ($ii$) \emph{the critical $\FMR$}, which is the highest $\FMR$ required to uphold that security level. Both indicators explicitly account for operational constraints to ensure their practical applicability.
The specific contributions are detailed below:

\vspace{-1mm}

\begin{enumerate}
\item[($i$)] \textbf{Untargeted Attack Analysis Based on \FMR.} This part focuses on the security implications of \emph{untargeted attacks} against biometric systems. 
Using the \FMR\ and its confidence interval, we derive the complexity of such attacks and their impact on system design, particularly in relation to the size of the user databases:
    \begin{itemize}
        \item We analyze the feasibility of untargeted attacks in biometric systems, focusing on a \textit{non-adaptive attacker}, as this approach enables the derivation of general recommendations applicable across different types of biometric systems and modalities. Adaptivity, in contrast, requires case-by-case studies specific to the system and modality, which are beyond the scope of this analysis.
        \item From this analysis, we derive the \textit{critical population} \textit{i.e.}, the maximum database size that guarantees a predefined security level against untargeted attacks.
        \item We introduce the \textit{critical \FMR\ }, the maximum acceptable \FMR\ required to maintain the predefined security level for a given database size.
        \item These results emphasize the relationship between \FMR, database size, and system security, providing practical recommendations for designing and deploying secure biometric systems at scale.
        \item The analysis is complemented by detailed graphs that illustrate the relationships between the parameters. 
    \end{itemize}
\item[($ii$)] \textbf{Analysis of the Biometric Birthday Problem.} This part delves into the biometric \textit{birthday problem}, exploring both its approximate and exact formulations. By integrating the confidence interval on the \FMR, we provide a refined analysis of collision probabilities and their impact on biometric system design:
\begin{itemize}
    \item 
    We provide a rigorous explanation of the approximate \textit{biometric birthday problem}, originally introduced by John Daugman~\cite{daugman2015information,daugman2016searching,Daugman24Understand}. His analysis relies on the strong independence assumption that the absence of a collision between a user \(A\) and another user \(B\) does not influence the probability of a collision between \(A\) and any other user \(C\). We refine his approximation by incorporating the confidence interval on the \FMR.  
    \item 
    From this analysis, we compute the approximate \textit{critical population}, where weak collisions---two different biometric templates being incorrectly matched, as in the \textit{birthday problem}---become significant, given a predefined security level and the confidence interval on the \(\FMR\).
    \item We define the \textit{critical $\FMR$} for weak collisions, offering guidance to reduce collision risks in large-scale deployments.  
    \item For the first time, we introduce the exact biometric \textit{birthday problem}, which overcomes the strong independence assumption made by Daugman. This analysis provides a precise estimation of collision probabilities and confirm the approximation of Daugman.  
    \item These results are complemented by detailed graphs visualizing collision probabilities, \textit{critical populations}, and \textit{critical \FMR\ } values, offering practical insights for system designers.
\end{itemize}
\end{enumerate}

\paragraph{Related Works}
Attacks on biometric systems are often analyzed on a case-by-case basis, reflecting the wide variety of secure biometric systems proposed in the literature. However, there has been limited exploration of generic attack methods or their effectiveness across diverse systems. An example of such generic attacks is \textit{targeted attacks}, where an adversary seeks to impersonate a specific individual within a database. Pagnin \textit{et al.}~\cite{PagninDAM14} introduced the "center search attack," a method that iteratively refines queries using minimal information leakage, such as a single bit indicating a match or mismatch, to locate specific templates. Their framework is built around binary templates and the Hamming distance as the underlying similarity metric. While this study provides valuable insights, its scope is constrained to specific template representations and similarity measures, limiting its relevance to more generalized biometric systems.

In contrast, \textit{untargeted attacks}, where an outsider aims to impersonate any user within a database, remain unexplored in the literature. 
Unlike targeted attacks, these attacks exploit the probabilistic nature of biometric matching, leveraging statistical properties such as the FMR to increase the likelihood of unauthorized access. 
The \FMR\ quantifies the probability that a random biometric input will match a stored template, making it a critical parameter for assessing the vulnerabilities of biometric systems. Furthermore, due to the non-uniform distribution of biometric data, \textit{untargeted attacks} are expected to perform significantly better than if the adversary were testing against uniformly distributed data. 
The increasing availability of biometric data, driven by advancements in generative artificial intelligence and synthetic data generation, further complicates these vulnerabilities.

Among the earliest tools for generating synthetic biometric data, SFinGe~\cite{cappelli2002synthetic}, 
has been widely used to generate synthetic fingerprint datasets, enabling the evaluation of fingerprint recognition algorithms under controlled conditions. 
Its ability to simulate diverse fingerprint patterns and variations has made it a valuable tool for testing and benchmarking in biometric research. Building on these foundations, Bahmani \textit{et al.}~\cite{bahmani2021high} proposed a method for generating high-fidelity fingerprints, focusing on improving quality, uniqueness, and privacy. 
These advancements have enhanced the utility of synthetic data for fingerprint recognition and system evaluation.

The emergence of generative adversarial networks (GANs) has further broadened the scope of synthetic biometric data generation. 
Engelsma \textit{et al.}~\cite{engelsma2022printsgan} introduced PrintsGAN, which leverages GANs to generate fingerprints with improved realism and diversity, addressing limitations in earlier methods. 
Beyond fingerprints, generative models have extended to other biometric modalities. 
Kim \textit{et al.}~\cite{kim2023dcface} proposed DcFace, a synthetic face generator based on a dual-condition diffusion model, capable of producing highly realistic and diverse facial images. These tools provide valuable resources for augmenting training datasets and testing machine learning models under varied conditions. While these synthetic data generators have advanced the state of biometric system evaluation, they also introduce potential security risks. The ease of creating realistic biometric data lowers the barrier for malicious actors to conduct attacks, such as exhaustive searches or the creation of forged templates. These developments underscore the importance of studying how system parameters, such as the FMR, interact with synthetic data and affect the robustness of biometric systems.

Biometric deployments at national or global scales introduce significant challenges, particularly regarding the \textit{biometric birthday problem}. 
Analogous to the classical \textit{birthday problem}, this problem quantifies the likelihood of collisions between biometric templates as the size of the database increases. 
For example, nationwide systems like Aadhaar in India and facial recognition initiatives in China~\cite{aiyar2017aadhaar,su2022facial}, each managing biometric data for over 1.4 billion individuals, face substantial collision risks that threaten uniqueness. 
Addressing these risks requires robust entropy management and careful parameterization, especially for large-scale deployments where we should maintain the balance between accuracy and security.

The \textit{biometric birthday problem}, first introduced by Daugman~\cite{daugman2003importance,daugman2015information,daugman2016searching,Daugman24Understand}, was briefly discussed in the context of iris recognition. 
Daugman’s work, while insightful, is limited in scope and relies on simplifying assumptions, such as the independence of collisions between templates. 
These assumptions, while practical for mathematical modeling, fail to capture structural dependencies and non-uniform distributions that often occur in real-world biometric systems. 
Building on Daugman’s contributions, we provide both an exact formulation of the \textit{biometric birthday problem} and a refined analysis of its approximate formulation. 
By removing the independence assumption, our study generalizes the problem to a wider range of biometric systems, addressing practical complexities often overlooked. 
This broader analysis provides a more precise understanding of collision risks in large databases and aims to offer practical guidance for mitigating these risks across diverse biometric modalities.

\paragraph{Outline}
The paper is structured as follows. Section~\ref{sec:Background} introduces the key notations and concepts necessary for understanding the subsequent sections. Section~\ref{sec:untar} presents the analysis of \textit{untargeted attacks}. 
Section~\ref{sec:coll} explores the problem of collisions in biometric systems. 
Finally, Section~\ref{sec:conc} concludes the paper.

\section{Preliminaries}
\label{sec:Background}
This section introduces the foundational concepts and theorems necessary for the remainder of this paper. First, we detail the calculation of the False Match Rate (\(\FMR\)) and discuss the construction and interpretation of its confidence intervals. Next, we present the attack models, focusing on both \textit{targeted} and \textit{untargeted} scenarios.  
Finally, we address the issue of collisions and detail their different types, including those related to the \textit{biometric birthday problem}.

\subsection{False Match Rate}
\label{FMR_def}
In the context of authentication systems, specifically in a one-to-one system, the predominant metric utilized is the False Match Rate ($\FMR$). 
This rate represents the probability of a biometric sample being incorrectly recognized by the matcher. 
In other words, this is the probability that the matcher incorrectly determines that a newly collected template matches the stored reference. 
The practical signification of the $\FMR$ lies in its impact on the security and usability of biometric systems. 
A high $\FMR$ implies a greater risk of unauthorized access, which is problematic for high-security applications and security-sensitive environments such as financial institutions or military facilities. 
As the $\FMR$ is directly linked to the security of the biometric system, it is important to quantify its value correctly.

\paragraph{Formal Definition.} According to the Face Recognition Technology Evaluation (FRTE) $1$:$1$ Verification~\cite{2024NIST1to1report}, given a vector of $n$ imposter scores $v$ and $T$ a threshold, an estimation of the $\FMR$ is
\begin{equation}
  \FMRe = \dfrac{1}{n}\sum\limits_{i=1}^n H(T-v_i) = \dfrac{\text{Number of false matches}}{\text{Total number of comparisons}}
\end{equation}
with $H(\cdot)$ the unit step function, and $H(0)$ taken to be $1$. $H(T-v_i)$ can only takes two values \textit{i.e.}, $0$ or $1$ meaning that either the corresponding pair of user match or not.

It is important to highlight that, based on the \FMR\ calculation method, the maximum security a system can achieve is limited to the inverse of the total number of comparisons. In other words, if we consider the entire human population (\textit{i.e.}, $8\times 10^{9}$) with one sample per individual, the smallest observable $\FMR$ is approximately $10^{-19.5}$. While this may seem counterintuitive, it underscores the fact that certain systems may have an accuracy level so sharp that it cannot be effectively measured using this method. For security-critical applications, it may be theoretically desirable to achieve an $\FMR$ this small; however, current technologies are not yet capable of reliably measuring or observing such a low $\FMR$. 
Indeed, the best-performing biometric systems, such as those tested by the National Institute of Standards and Technology (NIST), can achieve extremely low \FMR s. 
For example, state-of-the-art systems, including facial recognition technologies, have been recorded with \FMR s as low as $10^{-8}$ and even \(10^{-9.6}\), corresponding to a false match probability of one in $40$ billion comparisons~\cite{grother2012irex,2024NIST1to1report,grother2009performance}. On the other hand, widely used commercial systems, such as Apple's FaceID~\cite{FaceID} and TouchID~\cite{TouchID}, and Google's AndroidID~\cite{AndroideID}, typically report an $\FMR$ on the order of \(10^{-6}\).

\paragraph{Experimental Estimation and Confidence Interval}
As shown by Schuckers~\cite{schuckers2010computational}, it is possible to compute confidence interval in biometric metrics such as $\FMR$. 
An empirical estimate of the FMR, denoted by $\FMRe$, can be calculated by dividing the number of false matches by the number of matches tested for a given system. By convention, if we consider that matching is a success then, $H(T-v_i)\sim X_i$ a random variable which follows a Bernoulli distribution of probability $\FMR$. 
The above equation can be rewritten as \(\FMRe = \frac{1}{n} \sum_{i=1}^{N} X_i\).
By using the Central Limit Theorem (see Theorem~\ref{thm:CLT} in Appendix~\ref{sec:CLT:appendix} for more details), we deduce that $\FMRe$ follows a normal distribution of mean $\mu$ and variance \(\sigma^2/n\). The estimation of $\mu$ 
is given by $\overline{\mu} = \frac{1}{n} \sum_{i=1}^{n} X_i = \FMRe$.
The estimation of the standard deviation $\sigma$ 
is given by 
\[\overline{\sigma} = \sqrt{\dfrac{1}{n-1}\sum\limits_{i=1}^{n}\left(X_i-\FMR\right)^2} = \sqrt{\dfrac{n}{n-1}\FMRe\left( 1-\FMRe \right)}.\]
Using the confidence interval (see Theorem~\ref{thm:CI} in Appendix~\ref{sec:CI:appendix} for more details) on the normal distribution, we have $\FMR = \FMRe \pm c_\alpha\sqrt{\FMRe(1-\FMRe) / (n-1)}$ with $c_\alpha$ the quantile of the Student's distribution with $n-1$ degrees of freedom. 
Throughout the rest of the paper, we use the following notations for respectively the Confidence Interval Lower bound (CIL) and the Confidence Interval Upper bound (CIU):
($i$) $\CIL{n}=\FMRe - c_\alpha\sqrt{\frac{\FMRe(1-\FMRe)}{n-1}}$ and ($ii$) $\CIU{n}=\FMRe + c_\alpha\sqrt{\frac{\FMRe(1-\FMRe)}{n-1}}$.

\subsection{Attack Models}
This section provides a detailed description of the two main types of attacks on biometric systems, alongside an explanation of biometric collisions. 
The first type, known as \textit{targeted attacks}, focuses on impersonating a specific user in a database. 
The second, \textit{untargeted attacks}, targets any user within the database. 
Lastly, we address the concept of collisions within a biometric database to present a comprehensive overview of security-related concerns.

\paragraph{Targeted Attacks.} 
An external attacker (\textit{i.e.}, an individual not registered in the system) attempts to impersonate a specific user in the database. 
To achieve this, the attacker has access to all the information necessary for the attack, including the user ID and other relevant details. 
As previously noted, the probability of successfully impersonating a specific user is equal to the system’s $\FMR$. 
On average, the number of biometric attempts required for the attacker to succeed is $1 / \FMR$. 
The $\FMR$ should be minimized, ideally conforming to the standards of cryptographic systems, which today require at least $128$ bits for a moderate level of security~\cite{barker2015profile}, and in the near future will correspond to the minimum level~\cite{barker2020Futur}. 
This would correspond to an $\FMR$ smaller than $2^{-128}$. 
However, to the best of our knowledge, no existing biometric systems achieve this level of security. 
Moreover, \textit{targeted attacks} can be performed both \textit{online} and \textit{offline}. 
In \textit{online attacks}, the attacker interacts directly with a remote biometric system to test guesses. 
In \textit{offline attacks}, the attacker tests a chosen biometric data against a locally available obfuscated database, allowing simulations without interacting with a remote system.

\paragraph{Untargeted Attacks.} An external attacker seeks to impersonate any user within the database rather than a specific individual. 
As in \textit{targeted attacks}, the attacker has access to any additional information required to carry out the attack. 
This scenario is commonly associated with identification systems~\cite{dargan2020comprehensive,de2003biometric,ribaric2005biometric,wayman1997generalized,palma2022biometric,jain2007handbook,jain1996introduction}, which reject users who are not known to the system. 
To illustrate \textit{untargeted attacks}, consider a device (\textit{e.g.}, a smartphone) with a biometric recognition system (\textit{e.g.}, fingerprint recognition). 
A user may enroll multiple fingerprints or even those of other individuals on their device. 
Consequently, an attacker attempting to unlock the device does not need to match a specific enrolled biometric but only one of the enrolled templates. 
This makes \textit{untargeted attacks} significantly easier than targeted ones. 
\textit{Untargeted attacks}, however, are not limited to identification systems. 
They are equally applicable to authentication systems, where the attacker attempts to match a guessed template with a specific enrolled identity. 
In this mode, the attacker must test each guessed template against the enrolled biometric data of a claimed identity, introducing additional constraints compared to the identification mode. 
Specifically, in authentication systems, the attacker must know or guess the claimed identity (\textit{e.g.}, the username or account ID) associated with the enrolled biometric template yielding additional costs. 
As for the previous attack, \textit{untargeted attack} may be performed in both \textit{online} and \textit{offline} setup.


\paragraph{Collisions.} 
In the context of a biometric database containing \( N \) users, we distinguish between two types of collisions: \textit{strong collisions} and \textit{weak collisions}. 

\begin{itemize}
    \item A \textit{strong collision} occurs when a specific user in the database can be impersonated by another individual. Formally, for a user \( u \) with biometric template \( b_u \), a strong collision happens if there exists another user \( v \neq u \) such that the biometric matching function satisfies \( \text{Match}(b_u, b_v) = \text{True} \). The probability of a strong collision for a single user is equal to the \FMR , which quantifies the likelihood of incorrectly matching two distinct biometric templates.

    \item A \textit{weak collision} occurs when at least two users in the database can impersonate each other. Formally, for a database of \( N \) users with biometric templates \( b_1, b_2, \ldots, b_N \), a weak collision exists if there exist two distinct users \( i \neq j \) such that \( \text{Match}(b_i, b_j) = \text{True} \). The probability of weak collisions increases with the size of the database and is related to the ``\textit{biometric birthday problem}''~\cite{Daugman24Understand}.
\end{itemize}

The probability of weak collisions depends on the \FMR\ and the database size. A low weak collision probability indicates robust system accuracy, reducing risks of misidentification and impersonation by malicious users. Conversely, a high weak collision probability highlights vulnerabilities, particularly in large-scale databases. Estimating these probabilities through probabilistic modeling or simulations helps optimize system parameters and supports the design of secure, scalable biometric systems.

Ideally, both probabilities should be minimized to ensure robust security. While strong collisions can often be controlled by maintaining a sufficiently low \( \FMR \), the weak collision probability may become significant as the database size increases. In practical scenarios, a weak collision probability below \( 50\% \) is often considered sufficient to maintain acceptable performance and reliability, but additional measures may be required to address risks associated with strong collisions in sensitive applications.

\subsection{Experimental Setup}
The numerical evaluations were performed in \textit{R}~\cite{ihaka1996r} with a precision level of $200$ bits, utilizing the \textit{Rmpfr}~\cite{maechler2024package} package to ensure high accuracy. The results were subsequently stored in CSV files, which were then processed using \textit{Python 3}~\cite{python3} for graphical representation. The Python libraries \textit{NumPy}~\cite{harris2020array}, \textit{Pandas}~\cite{mckinney2010data}, and \textit{Matplotlib}~\cite{Hunter:2007} were employed to generate the visualizations. The selection of the parameter values was made in such a way as to both reflect the current state-of-the-art in biometric systems and provide insight into the performance levels required to achieve high security and scalability.

\section{Untargeted Attacker}
\label{sec:untar}

In this section, we address two key topics. First, we derive the complexity of \textit{untargeted attacks}, analyzing the resources required under various configurations. 
Second, leveraging this analysis, we compute two critical metrics: the maximum population size a system can securely accommodate at a specified \FMR , and the minimum $\FMR$ required to support a given population while maintaining a desired security threshold.

Our study focuses on the identification mode, where an attacker attempts to match a query against a large database. In the authentication mode, however, the complexity of attacks scales by a factor of \(N\), where \(N\) represents the total number of login attempts.
In this scenario, the attacker repeatedly uses the same guess for each login attempt before proceeding to a new guess. This adjustment is conceptually similar to \textit{password spray attacks}, where attackers exploit a large user base to attempt a single password guess per user, thus bypassing lockout mechanisms. 
Consequently, the \textit{critical population} size and \textit{critical $\FMR$} derived here are left to the reader for recalculation to reflect this adjustment. 
Each of these findings is supported by extensive experiments, with tables and graphs illustrating the results and complementing the derived theoretical bounds.


\subsection{Untargeted Attack Complexity}
\label{sec:untar-complexity}

%
%

We analyze the median number of attempts required for a \textit{non-adaptive attacker} to successfully gain acceptance in an identification system. In this context, the attacker aims to match one of the valid users from the database. This scenario can be modeled as follows: the attacker submits a query to the system, and, in this same round, this input is compared against all entries in the database.

To formalize this, let \( A_i \) represent the event that the attacker’s input matches the \( i \)-th valid user in the database. The probability of \( A_i \) is given by the system's \(\FMR\). To account for practical variability, our analysis incorporates a confidence interval on \(\FMR\), reflecting the uncertainty or fluctuation in this metric due to measurement errors, unavoidable sampling errors, or inadequacies in the sampling frame or reference database.

In this context, we consider a \textit{non-adaptive attacker}, which refers to an attacker that does not adapt its strategy based on the outcomes of previous attempts. This type of attacker submits queries independently across rounds, making it a general model applicable to any biometric system. While this assumption provides broad applicability, it is worth noting that in practice, for a specific biometric system and modality, an adaptive attacker could refine its guesses based on previous unsuccessful attempts. For example, the attacker might use feedback from the system to narrow down possible matches. However, such adaptive strategies depend on system-specific characteristics and are beyond the scope of this general analysis.

In this paper, we consider two distinct scenarios regarding the statistical dependence of the events \( A_i \) within the same round:

\begin{itemize}
    \item \textbf{Independence of events \( A_i \):} In this case, the events \( A_i \) are assumed to be mutually independent, meaning that the success of matching one user does not affect the probability of matching another user within the same round.

    \item \textbf{Dependence of events \( A_i \):} Here, we introduce the notion of \emph{collision dependence}, where the success of matching one user might influence the probability of matching others. It is important to note that this dependence is restricted to the same round, as we assume a \textit{non-adaptive attacker} throughout the paper. This implies no auxiliary information is gained or exploited between rounds, ensuring independence across rounds.
\end{itemize}

Under this framework, and noting that the probability of not matching user \( i \) is given by \( 1 - \FMR \), bounded within its confidence interval, we establish the following theoretical result regarding the median number of attempts required.

\begin{theorem}[Bounds on the Complexity of the Untargeted Attacker]
  \label{thm:Approx_untargeted}
  Let \(\FMR\) denote the False Match Rate of a biometric recognition system, with an estimated value \(\FMRe\) obtained from \(n\) comparisons at a significance level \(\alpha\). 
  The median number of attempts required for an untargeted attacker to successfully match one of the \(N\) users is bounded as follows:
  \begin{enumerate}
    \item Under the assumption of independence, the median number of attempts satisfies:
    \[
    \Omega\left(2^{-\log_2\left(\FMRe + \FMRe^2\right) - \log_2(N)}\right)
    \quad \text{and} \quad
    O\left(2^{-\log_2\left(\FMRe\right) - \log_2(N)}\right).
    \]

    \item Under the assumption of non-independence, and for \(\FMR \leq \frac{1}{2N}\), the bounds are:
    \[
    \Omega\left(2^{-\log_2(\FMRe) - \log_2(N^2\cdot\FMRe^2)}\right)
    \quad \text{and} \quad
    O\left(2^{-\log_2\left(\FMRe\right)}\right).
    \]
  \end{enumerate}
Accounting for the confidence interval, $\FMRe$ is replaced by $\CIL{n}$ in the lower bound and by \(\CIU{n}\) in the upper bound.
\end{theorem}

\begin{proof}
    \textbf{Case 1: Independence Assumption.} 
    Let \( A_i \) represent the event that the attacker does not match the \( i \)-th user. 
    Given a biometric system with performance \(\FMR\), the probability of \( A_i \) is \( 1 - \FMR \). 
    The events \( A_i \) are not necessarily independent, and by the chain rule of probability, we have:
    \[
    \mathbb{P}(A_1 \cap A_2 \cap \cdots \cap A_N) = \mathbb{P}(A_1) \prod_{i=2}^N \mathbb{P}(A_i \mid A_1 \cap A_2 \cap \cdots \cap A_{i-1}).
    \]
    Assuming independence between the events \( A_i \), the probability that the attacker does not match any user is:
    \[
    \mathbb{P}(A_1 \cap \cdots \cap A_N) = \prod_{i=1}^{N} (1 - \FMR) = (1 - \FMR)^{N}.
    \]
    Let \( V_i \) be the event that the attacker matches one of the \( N \) users at the \( i \)-th round. Since the attacker is non-adaptive, the sequence of \( V_i \) events are independent Bernoulli trials, and we are interested in the first success. 
    The median of a geometric distribution, representing the law of first success, is given by $m_{out} = \left\lceil\frac{-1}{\log_2(1 - p)}\right\rceil$.
    The probability of success in a single trial is given by $p = 1 - (1 - \FMR)^{N}$.
    Using the series expansion of \( \ln(1 - x) \) for \( 0 \leq x \leq 1/2 \), we have:
    \[
    \frac{-x - x^2}{\ln 2} \leq \log_2(1 - x) \leq \frac{-x}{\ln 2}.
    \]
    From this, we find that $\log_2(1 - p) = \log_2((1 - \FMR)^{N}) = N \log_2(1 - \FMR)$.
    Hence, we obtain the following bounds:
    \[
    \frac{-N \cdot (\FMR + \FMR^2)}{\ln 2} \leq \log_2(1 - p) \leq \frac{-N \cdot \FMR}{\ln 2},
    \]
    which lead to:
    \[
    \frac{\ln 2}{N} (\FMR + \FMR^2)^{-1} \leq m_{out} \leq \frac{\ln 2}{N} \FMR^{-1}.
    \]
    Then, considering the fact that \(\FMR = \FMRe \pm c_\alpha \sqrt{\frac{\FMRe(1 - \FMRe)}{n - 1}}\), where \(n\) is the number of comparisons used to estimate \(\FMRe\) and \(c_\alpha\) is the quantile of the Student's distribution for significance level \(\alpha\), the result follows.
  
    \textbf{Case 2: Non-Independence.} 
    The Bonferroni inequality~\cite{bonferroni1936teoria} provides an upper bound for the probability of the union of events. Given \( N \) events \( A_1, A_2, \dots, A_N \), the inequality states:
  \begin{eqnarray*}
  \mathbb{P}(A_1 \cap A_2 \cap \cdots \cap A_N) & \ge & 1 - N + \sum_{i=1}^N \mathbb{P}(A_i).
  \end{eqnarray*}
  Moreover, note that we have for all $1 \leq i \leq N$:
  \begin{eqnarray*}
    \mathbb{P}(A_1 \cap A_2 \cap \cdots \cap A_N) & \leq & \mathbb{P}(A_i).
    \end{eqnarray*}
  In this case, applying the Bonferroni inequality and the above inequality gives the following bounds:
  \begin{eqnarray*}
    1-N\cdot\FMR \leq & 1-p & \leq 1-\FMR\\
    \log_2(1-N\cdot\FMR) \leq & \log_2(1-p) & \leq  \log_2(1-\FMR)
  \end{eqnarray*}
  if $\FMR \leq 1/N$. 
  Then, using the expansion series, we have:
  \begin{eqnarray*}
    \dfrac{-N\cdot\FMR-N^2\cdot\FMR^2}{\ln 2} \leq & \log_2(1-p) & \leq \dfrac{-\FMR}{\ln 2}\\
    \dfrac{\ln 2}{N\cdot\FMR+N^2\cdot\FMR^2} \leq & m_{out} & \leq \dfrac{\ln 2}{\FMR}
  \end{eqnarray*}
  if $\FMR \leq 1/2N$.
  Then, considering the fact that \(\FMR = \FMRe \pm c_\alpha \sqrt{\frac{\FMRe(1 - \FMRe)}{n - 1}}\), where \(n\) is the number of comparisons used to estimate \(\FMRe\) and \(c_\alpha\) is the quantile of the Student's distribution for significance level \(\alpha\), we have:
  \begin{equation*}
    \dfrac{\ln 2}{N\cdot\CIU{n}+N^2\cdot\CIU{n}^2} \leq m_{out} \leq \dfrac{\ln 2}{\CIL{n}}
  \end{equation*}
  and the result follows.
\end{proof}


The results presented in Theorem~\ref{thm:Approx_untargeted} show that the security of biometric systems is primarily determined by two factors: the False Match Rate (\FMRe) and the number of users \(N\). 
It is clear that as expected, systems with lower \FMR\ and fewer users are inherently more secure. 

When analyzing the lower bounds, the first term in the exponent is $-\log_2\left(\CIU{n} + \CIU{n}^2\right)$ for the first case and $-\log_2\left(\CIU{n}\right)$ for the second case. Considering only this first term, the median number of rounds required for an attacker is smaller in the first case than in the second. 
Examining the second term, $-\log_2(N)$ in the first case and $-\log_2(N^2\CIU{n})$ in the second case, we observe that when $\CIU{n} \leq 1/N$ (a scenario likely to hold as in this case $\FMR \leq 1/2N$), the number of rounds for the first case remains smaller than that of the second. 
Consequently, for the lower bound, the first case grants the attacker a greater advantage, requiring fewer attempts to succeed. Conversely, for the upper bound, it is clear that the number of rounds in the first case is consistently smaller than in the second. 
Therefore, when seeking to establish a pessimistic security bound, it is appropriate to rely on the lower bound derived from the first case. 

\begin{rem}
No matter which specific case is considered, asymptotic bounds on the median number of rounds required for an untargeted attacker to succeed are as follows:
  \[
  \Omega\left(2^{-\log_2\left(\FMRe + \FMRe^2\right) - \log_2(N)}\right)
  \quad \text{and} \quad
  O\left(2^{-\log_2\left(\FMRe\right)}\right).
  \]
Incorporating the confidence interval, $\FMRe$ is replaced by $\CIL{n}$ in the lower bound and by \(\CIU{n}\) in the upper bound.
\end{rem}

\subsection{Critical Population for the Untargeted Attack}
In this section, we derive the \textit{critical population} size for a given biometric system. 
Specifically, we establish the maximum number of users that a database can securely manage while satisfying both security constraints and system parameters, such as the \FMR\ and the collision probability.


\begin{theorem}[Critical Population Against Untargeted Attacker]
  \label{cor:popUA}
  Consider a biometric system operating with an error rate \(\FMR\), estimated as \(\FMRe\) based on \(n\) comparisons at a significance level \(\alpha\). The \textit{critical population} \(N\), such that the median number of rounds required by an untargeted attacker to successfully impersonate any user is smaller than a given threshold \(S\), is given by:
  \[
    N = \frac{\ln 2}{S \left( \FMRe + \FMRe^2 \right)}.
  \]
  When accounting for the confidence interval, \(\FMRe\) is replaced by the upper bound \(\CIU{n}\).
\end{theorem}

\begin{proof}
  Let $S$ be the security desired (\textit{e.g.}, $2^{128}$) for a biometric system with a performance estimated by $\FMRe$ as before. Thus, using the smallest lower bound from Theorem~\ref{thm:Approx_untargeted} (\textit{i.e.}, the lower bound in the first case), we require:
  \begin{eqnarray*}
    S & \leq & \dfrac{\ln 2}{N\cdot\left(\CIU{n}+\CIU{n}^2\right)}\\
  \end{eqnarray*}
  which may be rewritten as
  \begin{eqnarray*}
  N & \leq & \dfrac{\ln 2}{S\left(\CIU{n}+\CIU{n}^2\right)}\\
\end{eqnarray*}
and the result follows.
\end{proof}

The theorem shows that the \textit{critical population} size \( N \) is inversely proportional to the security threshold \( S \) and the term \(\CIU{n} + \CIU{n}^2\). 
Increasing \( S \) enhances system security by requiring more attempts for an attacker to successfully impersonate a user. However, this comes at the expense of a reduced \( N \), imposing stricter limits on the system's scalability. 
While higher values of \( S \) may be offset by reducing \(\CIU{n} + \CIU{n}^2\), practical constraints on the \(\FMR\) limit this compensation. 

\begin{rem}
  If the \textit{critical population} size \( N \) falls below \( 1 \), the desired security threshold \( S \) is unattainable with the system's current accuracy, indicating the need for significant improvements in system performance or a relaxation of the security requirements. 
  When \( N \) lies between \( 1 \) and \( 2 \), the system ensures adequate security for single-user scenarios and should be consistent with the targeted attack. 
  Notably, if \( N \) ranges from \( 3 \) to \( 10 \), the system is capable of securely managing small-scale deployments. 
  Such scenarios are typical in consumer devices like smartphones, where biometric systems frequently support the enrollment of \( 5 \) to \( 10 \) fingerprints for user authentication or identification. 
\end{rem}

\subsection{Critical $\FMRe$ for the Untargeted Attack}

As shown in the previous section, the \FMR\ plays a key role in ensuring that the median number of attempts required for an attacker to succeed aligns with the desired security level for a system supporting a given number of users.  
In this section, we shift our focus to determining the necessary error rate (\(\FMR\)) to ensure that a biometric system managing \(N\) users achieves \(S\) bits of security against \textit{untargeted attacks}.  
This analysis provides practical guidelines for configuring system performance parameters, particularly for large-scale deployments where security and scalability must be carefully balanced.

\begin{theorem}[Critical Estimated $\FMR$]
  \label{cor:fmr_untargeted}
Given a database with \(N\) users and a security level \(S\) against the \textit{untargeted attack}, the biometric error rate \(\FMR\), estimated as \(\FMRe\) from \(n\) comparisons with a significance level \(\alpha\), must satisfy:
  $\FMRe \leq \sqrt{\frac{1}{4}+\frac{\ln 2}{N\cdot S}}-\frac{1}{2}$.
Accounting for the confidence interval, $\FMRe$ is replaced by $\CIU{n}$.
\end{theorem}


\begin{proof}
Using Theorem~\ref{thm:Approx_untargeted}, the security level \(S\) satisfies 
\(S \leq \frac{\ln 2}{N (\FMR + \FMR^2)}\), leading to the quadratic inequality 
\(\FMR^2 + \FMR - \frac{\ln 2}{N \cdot S} \leq 0\). 
Studying the function \(f(x) = x^2 + x - \frac{\ln 2}{N \cdot S}\) and solving for its roots, we find 
\(\FMR = \frac{-1 \pm \sqrt{1 + \frac{4 \ln 2}{N \cdot S}}}{2}\). 
Since \(\FMR \geq 0\), we select the positive root, giving 
\(\FMR \leq \sqrt{\frac{1}{4} + \frac{\ln 2}{N \cdot S}} - \frac{1}{2}\). 
Hence, the result follows.
\end{proof}

The inequality shows that the upper confidence interval for the $\FMRe$, \(\CIU{n}\), is fundamentally constrained by both the desired number of users, \(N\), and the desired security level, \(S\). 
As the product \(N \cdot S\) increases, \(\CIU{n}\) diminishes. 
Importantly, \(N\) and \(S\) have equivalent influence on reducing \(\CIU{n}\); neither parameter dominates the other, as both contribute equally to the requirement for a smaller \(\FMRe\). 
Consequently, achieving higher \(N\) or higher \(S\) necessitates reducing \(\CIU{n}\), which depends on improving the \(\FMR\) and the statistical confidence of its estimation.

\begin{rem}[Confidence Paradox]
  It is important to emphasize the problem associated with obtaining a tight confidence interval on the \(\FMR\). For the upper bound \(\CIU{n}\) to be meaningful, it must be smaller than the right-hand side of the inequality in Theorem~\ref{cor:fmr_untargeted}. Simple approximations reveal that the number of comparisons, \(n\), required to achieve this is roughly proportional to \(N \times S\), where \(N\) is the population size and \(S\) the desired security level. 

  This dependency leads to an exponential growth in the number of comparisons as \(N\) or \(S\) increases. For example, a system designed to handle \(1000\) users while achieving \(50\) bits of security would require approximately \(2^{498}\) or \(10^{150}\) comparisons, which is computationally infeasible.

  Consequently, in such cases, confidence intervals on the \(\FMR\) must be disregarded. This limitation implies that no guarantees on the required \(\FMR\) can be reliably established, regardless of the chosen confidence level \(\alpha\). 
\end{rem}

\vspace{-3mm}

\subsection{Numerical Evaluations}

This section presents a detailed numerical evaluation of critical parameters for biometric systems under untargeted attacks. We estimate the \(\FMR\) required to maintain a specified level of security for a fixed population size, \(N\). Note that, confidence intervals on the \(\FMR\) are omitted, in alignment with earlier observations (see Appendix~\ref{Annexe:ICimpact} for an insight on the impact of the confidence interval).

Most contemporary systems operate with an \(\FMR\) around \(10^{-6}\). 
Figure~\ref{fig:critical_fmr_untar} provides further insights. 
In the absence of specific security constraints, an attacker making four attempts to impersonate an individual limits the maximum manageable population size to approximately \(10^6\). 
For small databases with only $10$ users, the maximum achievable security level is $10$ bits (\(2^{10}\) attempts required to compromise the system). 

To meet cryptographic standards, systems must achieve at least $112$ bits of security, with $128$ bits becoming mandatory from $2030$ onward~\cite{barker2015profile,barker2020Futur}. 
For small-scale databases (\textit{e.g.}, $10$ users and $112$ bits of security), achieving this standard requires an \(\FMR \leq 10^{-35}\). 
For large-scale systems accommodating \(10^9\) users and requiring $128$ bits of security, the required \(\FMR\) decreases further to \(10^{-45}\). 

These results underscore the significant challenges biometric systems face in balancing scalability and cryptographic security. 
Achieving such low \(\FMR\) values poses technical and computational hurdles, emphasizing the importance of advancing biometric technologies to meet these stringent requirements.


\begin{figure}
  \centering
  \includegraphics[width=\columnwidth]{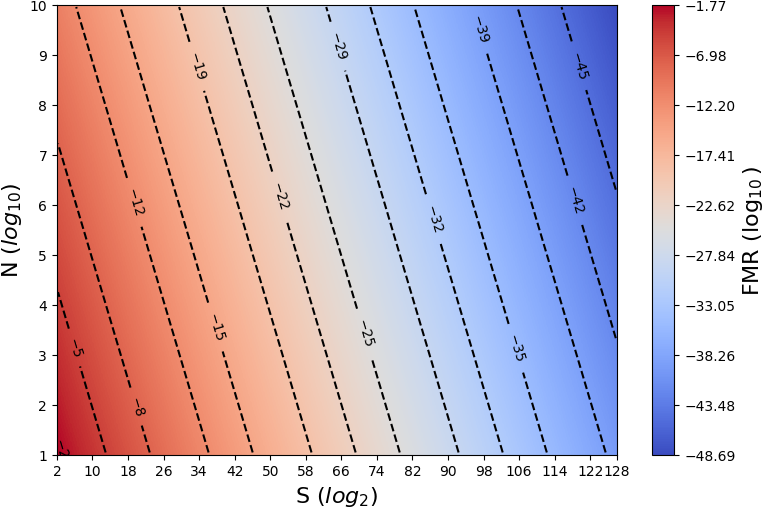}
\caption{Critical $\FMR$ as a function of the population size $N$ and the security level $S$ against \textit{untargeted attacks}.}
  \label{fig:critical_fmr_untar}
  \vspace{-3mm}
\end{figure}

\section{On the Weak Collisions}
\label{sec:coll}
This section investigates the occurrence of collisions in biometric identification systems. 
Collisions are particularly problematic as they increase the probability of user misidentification, compromising the reliability of biometric-based authentication. 
This issue is especially critical in access control systems with hierarchical privileges, where a collision could lead to unauthorized access to sensitive areas or information.

We begin by analyzing the probability of collisions in both approximate and exact cases. 
For the approximate case, we build upon Daugman’s framework~\cite{Daugman24Understand}, providing a rigorous mathematical estimation of collision probabilities. 
For the exact case, we derive the collision probability through detailed probabilistic modeling. 
Additionally, we identify the \textit{critical population} size at which the collision probability surpasses a defined threshold. 
Furthermore, we derive the necessary \(\FMRe\) (\textit{Critical \FMR}) to ensure a given maximal probability of collision for a system with \(N\) users. 
Each of these points is supported by extensive experiments, with tables and graphs illustrating the results and complementing the derived theoretical bounds.

\subsection{Approximated Biometric Birthday Problem}
\label{sec:Approx_bith}

In this section, we derive an approximation for the probability of a collision in a biometric database, leveraging the concept of the \emph{biometric birthday problem}, as proposed by various authors~\cite{daugman2003importance,daugman2015information,daugman2016searching,Daugman24Understand}.

The \textit{biometric birthday problem} extends the classical \textit{birthday problem} (see~\cite{mckinney1966generalized,Frank1991Birthday,Borja2007Birthday} for more detailed) to biometric systems and is formally stated as follows: \emph{"Given a biometric system with a known \(\FMR\), how many randomly selected individuals must be considered before the probability of at least one pair experiencing a biometric collision (\textit{i.e.}, a false match) exceeds a given threshold p?"}

\subsubsection{An Approximation for the Biometric Birthday Problem.}
First, we revisit Daugman's result~\cite{Daugman24Understand} (Theorem~\ref{Daugmanthm}) and provide a rigorous proof to substantiate his claims.

\begin{theorem}[Approximated Biometric Birthday Problem]
  \label{Daugmanthm}
  Given a biometric system operating with an error rate $\FMR$, an approximation of the probability that there is at least one collision in a database of $N$ user is given by $\mathbb{P}(K) \approx 1-(1-\FMR)^{\Ncomb}$ with $\Ncomb=\frac{N(N-1)}{2}$.
\end{theorem}

\begin{proof}
Let \(M_i\) represent the event "The individuals within the \(i\)-th pair of \(\overline{\mathcal{P}}\) do not falsely match," \(Q(N)\) the probability of the event \emph{"There is no false match among \(N\) pairs,"} and \(\mathbb{P}(N)\) the probability of the event \emph{"There is at least one false match among \(N\) pairs."} We have \(\mathbb{P}(M_i) = \FMR\) and \(\mathbb{P}(K) = 1 - Q(K)\).

Let \(\overline{\mathcal{P}}\) denote the set of all possible unordered pairs among \(N\) individuals. 
Then, the total number of pairs of individuals is given by $|\overline{\mathcal{P}}| = \dbinom{N}{2}  = \dfrac{N!}{2!(N-2)!} = \dfrac{N(N-1)}{2} = \Ncomb$.
Then, we have:
\begin{eqnarray*}
    Q(N) = \mathbb{P}\left(\bigcap\limits_{i=1}^{|\overline{\mathcal{P}}|} M_i\right) & = & \mathbb{P}\left(M_1\right)\times \prod\limits_{i=2}^{\Ncomb}\mathbb{P}\left(M_i \big |\bigcap\limits_{j=1}^{i-1} M_j\right) \\
  & \approx & \mathbb{P}\left(M_1\right)\times \prod\limits_{i=2}^{\Ncomb}\mathbb{P}\left(M_i\right).\\
\end{eqnarray*}
The above approximation is derived by assuming that \(M_i\) is independent of \(M_j\) for all \(i > j\), which is an important assumption in the result. 
It can be deduced that $Q(N) \approx (1-\FMR)^{\Ncomb}$ by using the fact that $\forall i \in \lbrace 1,\dots,\Ncomb\rbrace$, $\mathbb{P}\left(M_i\right) = 1-\FMR$.
Then, we have that $\mathbb{P}(K) \approx 1-(1-\FMR)^{\Ncomb}$ which explains and correspond to the result given by Daugman~\cite{Daugman24Understand} for the \textit{biometric birthday}.
\end{proof}

Then, Theorem~\ref{thm:Approx_bio} extend Daugman's results by incorporating confidence bounds on the $\FMR$, enabling a more accurate estimation of this probability.

\begin{theorem}[Approximated Solution for the Biometric Birthday Problem]
  \label{thm:Approx_bio}
  Given a biometric system operating with an error rate $\FMR$ estimated by $\FMRe$ on $n$ comparisons with a significance level $\alpha$, an approximation of the probability that there is at least one collision in a database of $N$ user is between
  \begin{equation*}
    \Omega\left(1-\left(1-\FMRe\right)^{\Ncomb}\right) \text{ and } O\left(1-\left(1-\FMRe\right)^{\Ncomb}\right)
  \end{equation*}
  with $\Ncomb=\frac{N(N-1)}{2}$.
Accounting for the confidence interval, $\FMRe$ is replaced by $\CIL{n}$ in the lower bound and by \(\CIU{n}\) in the upper bound.
\end{theorem}

\begin{proof}
  Using Theorem~\ref{Daugmanthm} and Section~\ref{FMR_def}, we have $\mathbb{P}(K)  \approx 1-\left(1-\FMRe \pm c_\alpha\sqrt{\dfrac{\FMRe(1-\FMRe)}{n-1}}\right)^{\Ncomb}$ with $n$ the number of 
comparisons used to estimate the approximated $\FMR$ (\textit{i.e.,} $\FMRe$) and $c_\alpha$ the quantile of the Student's distribution with $n-1$ degrees of freedom. The result follows.
\end{proof}

As the number of pairwise comparisons, \(\Ncomb = N(N-1)/2\), increases quadratically with the number of users \(N\), the exponential terms \((1-\CIL{n})^{\Ncomb}\) and \((1-\CIU{n})^{\Ncomb}\) rapidly converge to zero as \(N\) grows. 
Consequently, the probability of a collision converges exponentially towards $1$ when the number of user increases. 
This confirm the intuition that larger databases are inherently more susceptible to collisions, underscoring the importance of a small \(\FMR\) in systems with a high number of user.

\subsubsection{Approximated Critical Population}

%
In this section, we investigate the notion of \textit{critical population}, which refers to the maximum number of users that a biometric system can support while maintaining a specified \(\FMR\) and collision probability. This concept provides practical guidelines for system design where collision risks must be mitigated.

Theorem~\ref{thm:Critical_bio} formalizes this relationship and establishes a rigorous framework for analyzing the trade-off between security requirements and user capacity.

\begin{theorem}[Approximated Critical Population]
  \label{thm:Critical_bio}
Given a biometric system operating with an error rate \(\FMR\), estimated as \(\FMRe\) from \(n\) comparisons at a significance level \(\alpha\), the maximum number of users \(N\) such that the probability of a collision is smaller than a chosen threshold \(p\) is $N \approx \sqrt{\dfrac{2\ln(1-p)}{\ln\left(1-\CIL{n}\right)}}$.
\end{theorem}

\begin{proof}
  If we define $p$ as the maximum probability that a collision occur, the following inequality can be written using Theorem~\ref{thm:Approx_bio}:
  \begin{align*}
    1-\left(1-\CIL{n}\right)^{\Ncomb} & \leq p \\
    N^2 - N - 2\dfrac{\ln(1-p)}{\ln\left(1-\CIL{n}\right)} & \leq 0.\\
  \end{align*}
  The study of the above function in $N$ yields 
  \begin{align*}
    N \leq \dfrac{1}{2}+\dfrac{1}{2}\sqrt{1+\dfrac{8\ln(1-p)}{\ln\left(1-\CIL{n}\right)}} \approx \sqrt{\dfrac{2\ln(1-p)}{\ln\left(1-\CIL{n}\right)}}\\
  \end{align*}
  and the result follows.
\end{proof}


Theorem~\ref{thm:Critical_bio} provides an approximation for the \textit{critical population} size, \(N\), given a target collision probability, \(p\), and the lower confidence bound \(\CIL{n}\) on the \(\FMR\). 
This result emphasizes the relationship between collision probabilities, system accuracy, and user capacity. 

It is worth noting that \(\FMR\) and \(p\) have opposite effects on the size of \(N\).
As \(p\) decreases, reflecting tighter security requirements, the \textit{critical population} \(N\) diminishes. 
Conversely, looser requirements on \(p\) allow \(N\) to grow. 
On the other hand, reducing \(\FMR\) allows the system to accommodate a larger population, while higher \(\FMR\) values result in a decrease in population size. 
The delicate balance between security requirements and system scalability is underscored by this duality.

\subsubsection{Required $\FMRe$ for Target Population Sizes}
A complementary approach to the collision problem involves reversing the perspective.
Instead of determining the \textit{critical population} size \(N\) for a given \(\FMRe\), the question becomes:
for a target population size \(N\), what upper bound on \(\FMRe\) is required to ensure that the probability of a collision remains below a specified threshold \(p\)? 
This alternative perspective offers practical insights into the required system accuracy for handling large user bases at a given probability of collision. As a result, we formalize this relationship in Theorem~\ref{thm:Critical_FMR}.
\begin{theorem}[Approximated Critical \FMRe]
  \label{thm:Critical_FMR}
  Given a database with $N$ user and $p$ the chosen maximal probability of a collision, the biometric error rate $\FMR$ estimated by $\FMRe$ on $n$ comparisons must be such that $\FMRe \leq 1-e^{\ln(1-p)/\Ncomb}$.
\end{theorem}


\begin{proof}
Using the same notations as above and Theorem~\ref{thm:Approx_bio}, we start with \( 1 - (1 - \CIL{n})^{\Ncomb} \leq p \). Taking the natural logarithm and rearranging terms, we obtain \( \ln(1 - \CIL{n}) \geq \frac{\ln(1-p)}{\Ncomb} \). By exponentiating both sides, we get: $\CIL{n} \leq 1 - e^{\frac{\ln(1-p)}{\Ncomb}}$.
By observing that \(\FMRe \leq \CIL{n}\), this implies:
$\FMRe \leq 1 - e^{\frac{\ln(1-p)}{\Ncomb}}$.
Yielding \( \FMRe \approx 1 - e^{\frac{\ln(1-p)}{\Ncomb}} \) for the maximal \(\FMRe\) needed if \(n\) is large enough, and the result follows.
\end{proof}

This result provides an upper bound on the estimated \(\FMRe\) required to maintain a collision probability below a specified threshold \(p\) for a given database size \(N\). 
The exponential term \(\exp\left(\frac{\ln(1-p)}{\Ncomb}\right)\) emphasizes the relationship between the collision probability, the population size, and the \(\FMRe\). 
As the number of pairwise comparisons \(\Ncomb = \frac{N(N-1)}{2}\) increases, the required \(\FMRe\) becomes progressively smaller, highlighting the strong accuracy requirements for systems handling large user bases. 
For smaller values of \(p\), reflecting higher security demands, \(\FMRe\) must decrease further to ensure the collision probability remains acceptably low.

\subsubsection{Numerical Evaluations}
To prevent weak collisions in a database with a given probability, we investigate the required \(\FMR\) under specific design constraints. 
If the critical population size is to be derived from the \(\FMR\) and the target collision probability, or conversely, if the collision probability must be determined for a given population size and \(\FMR\), these relationships are conveniently illustrated in the filled contour plot of the figure~\ref{fig:Critical_FMR}.
To avoid redundancy, we focus below on selected numerical evaluations of particular interest.

%
%
First, let us recall that most contemporary systems operate with a false match rate (\(\FMR\)) around \(10^{-6}\). Considering the classical \textit{birthday problem}, where the collision probability is set to \(1/2\), such systems can accommodate approximately one hundred users.

In contrast, large-scale systems, such as those deployed in countries like China, must handle populations approaching \(10^9\) users. To maintain a low collision probability in such cases, the required \(\FMR\) would need to be as low as \(10^{-18}\). Considering the lowest known \(\FMR\) values (approximately \(10^{-9}\)), current large-scale biometric systems inherently face a significant number of collisions.

An ideal biometric system should minimize the collision probability while supporting populations up to \(10^{10}\) users to account for the global human population and effectively mitigate misidentification risks. Achieving this ambitious target would require an \(\FMR\) of \(10^{-20}\), which represents a significant technological challenge for future system designs.


\begin{figure}
  \centering
  \includegraphics[width=\columnwidth]{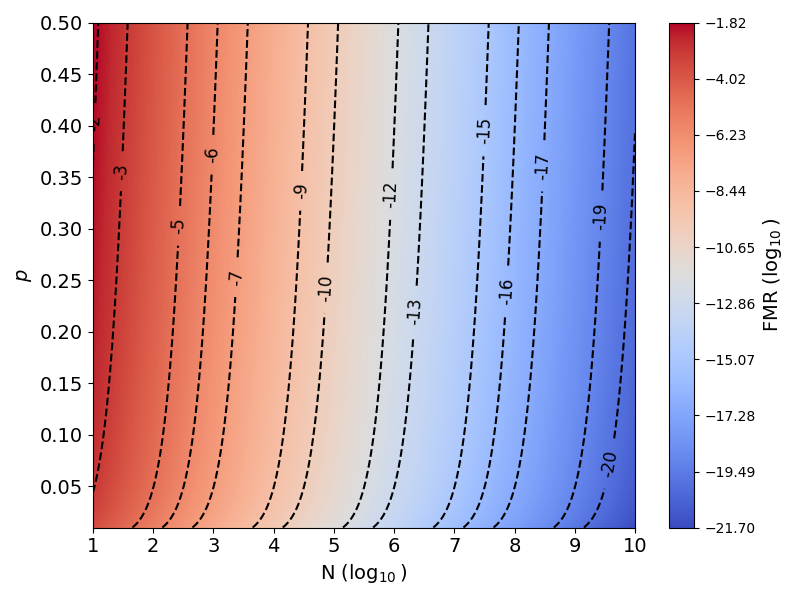}
\caption{The maximal $\FMRe$ (in $\log_{10}$) required to handle a population size $N$ with a maximum allowed probability $p$ for a collision occurrence.}
  \label{fig:Critical_FMR}
  \vspace{-4mm}
\end{figure}

\subsection{The Exact Biometric Birthday Problem}
While the result provided in Section~\ref{sec:Approx_bith} gives an estimation, it is possible to compute the exact probability of the occurrence of a weak collision. 
\begin{theorem}
  \label{thm:real_bio_bith}
  Given a database with \(N\) users and operating with a False Match Rate estimated as \(\FMRe\) from \(K\) comparisons, the probability that there is at least one weak collision is:
  $$\mathbb{P}(N) = 1-(1-\FMR)\dbinom{(1-\FMR)\Kcomb-1}{\Ncomb}\bigg/\dbinom{\Kcomb-1}{\Ncomb}.$$
\end{theorem}


\begin{proof}
  Let $K$ denote the number of users used to compute de $\FMRe$. Recall that we have $$\FMRe = \dfrac{\text{Number of pairs which falsely match}}{\text{Total number of pairs}}$$ where the total number of pairs is $\Kcomb=\dfrac{K(K-1)}{2}$, $$\FMR = \FMRe \pm c_\alpha\sqrt{\dfrac{\FMRe(1-\FMRe)}{\Kcomb}}$$ and  
\begin{eqnarray*}
  Q(N) & = & \mathbb{P}\left(\bigcap\limits_{i=1}^{\Ncomb} M_i\right) 
  = \mathbb{P}\left(M_1\right)\times \prod\limits_{i=2}^{\Ncomb}\mathbb{P}\left(M_i \big |\bigcap\limits_{j=1}^{i-1} M_j\right).\\
\end{eqnarray*}
Note that the probability \(\mathbb{P}(M_{i+1} \mid \bigcap_{j=1}^{i} M_j)\) is the ratio of the number of remaining pairs that do not falsely match to the total number of remaining pairs. Specifically, the number of remaining pairs that do not falsely match is \(\Kcomb(1-\FMR) - i + 1\), while the total number of remaining pairs is \(\Kcomb - i + 1\).
Then, we have:
\begin{eqnarray*}
  Q(N) & = & \mathbb{P}\left(M_1\right)\times \prod\limits_{i=2}^{\Ncomb}\mathbb{P}\left(M_i \big |\bigcap\limits_{j=1}^{i-1} M_j\right)\\
  & = & (1-\FMR)\prod_{i=2}^{\Ncomb} \dfrac{(1-\FMR)\Kcomb-i+1}{\Kcomb-i+1}\\
  & = & (1-\FMR)\dbinom{(1-\FMR)\Kcomb-1}{\Ncomb}\bigg/\dbinom{\Kcomb-1}{\Ncomb}\\
\end{eqnarray*}
and the result follows.
It is important to note that if $(1-\FMR)\Kcomb-1 \leq \Ncomb$ then, $Q(N)=0$, as there are not enough remaining pairs to satisfy the condition.
\end{proof}




The collision probability derived in Theorem~\ref{thm:real_bio_bith} is influenced by three key parameters: the False Match Rate (\(\FMR\)), the desired number of users (\(N\)), and the number of users used for estimating the \(\FMR\) (\(K\)).

As the \(\FMR\) increases, the probability of a collision also increases. Conversely, as \(\FMR\) approaches zero, both the upper and lower bounds of the collision probability tend to zero, indicating a lower probability of collision with higher system accuracy.

It is important to note that when \(\FMR\) is nonzero, and given that it is computed based on \(K\) users, as \(N\) approaches \((1-\FMR)K\), the collision probability increases accordingly. 
Informally, this can be explained by the fact that the number of pairs not in a false match decreases. Thus, the probability of drawing a pair in a false match increases with \(N\). 
Once all non-colliding pairs are exhausted, only colliding pairs remain, and the probability of a collision becomes one (see pigeonhole principle~\cite{trybulec1990pigeon}).
As a result, smaller databases tend to result in a lower collision probability, consistent with earlier observations on the relationship between database size and security.

Finally, it is instructive to examine the behavior as \(K \to \infty\), where the number of users used to compute the \(\FMR\) becomes arbitrarily large. 
In this case, the collision probability approaches \(1-(1-\FMR)^{\mathcal{N}}\), which is the result derived in section~\ref{thm:Approx_bio}. To be convinced, recall that $Q(N)=(1-\FMR)\prod_{i=2}^{\Ncomb} \dfrac{(1-\FMR)\Kcomb-i+1}{\Kcomb-i+1}$ (see proof of Theorem~\ref{thm:real_bio_bith}). Then, as \(K \to \infty\), this is equivalent to $Q(N)=(1-\FMR)\prod_{i=2}^{\Ncomb}(1-\FMR)$ implying the result.

\vspace{-3mm}

\section{Conclusion}
\label{sec:conc}
In this paper, we introduced the concept of \textit{untargeted attacks}, wherein an adversary aims to impersonate any user within a database. 
We analyzed the complexity of such attacks, leading to the definition of the \textit{critical population} (\textit{i.e.}, the maximum allowable database size) for a specified security level $S$ and an estimated system performance $\FMRe$. 
Furthermore, we derived the required $\FMRe$ to achieve a desired security level for a given database size. 
Our experimental results reveal that current systems achieving $\FMRe \approx 10^{-6}$ can offer at most $20$ bits of security against \textit{untargeted attacks} and fewer than $10$ bits of security when the database exceeds $10$ users. 
To ensure $128$ bits of security for a system accommodating the global human population (\textit{i.e.}, approximately $10^9$ users), an $\FMRe$ smaller than $10^{-45}$ is necessary.

Additionally, we investigated the biometric \textit{birthday problem} to study collision probabilities within a database. 
We determined both approximate and exact probabilities of weak  collisions as functions of $\FMRe$ and database size. 
From this analysis, we derived expressions for the \textit{critical population} in terms of $\FMRe$ and collision probability, as well as the required $\FMRe$ for a given database size to maintain a specific collision probability. 
Experimental results indicate that systems with $\FMRe \approx 10^{-6}$ can handle up to $100$ users before the collision probability exceeds $50\%$. 
For databases accommodating the global human population, the $\FMRe$ must be smaller than $10^{-18}$ to keep the collision probability below $50\%$.

In summary, this work provides an in-depth analysis of the $\FMR$ metric, revealing that biometric systems often considered secure may, in fact, fall short of expected security levels and more generally, to prevent both untargeted attacks and collisions, large scale systems must ensure $\FMRe\leq 10^{-45}$.

The present work concludes that protecting biometric data relies on two important aspects.  
First, the design of truly secure (\textit{zero-knowledge}) authentication protocols. Such protocols inherently satisfy the model of \textit{minimal information leakage}. However, poorly parameterized systems or low entropy in biometric templates can still undermine security.

So, the second aspect is the enhancement of biometric data entropy. Higher entropy directly strengthens resistance to exhaustive search attacks, reducing the likelihood of successful \textit{untargeted attacks}. Without sufficient entropy, adversaries can perform exhaustive searches to recover the original biometric features or approximate them, much like cracking weak password hashes. Techniques such as the fusion of multiple biometric modalities or the randomization of biometric data using a pseudo-random secret or a strong password can significantly mitigate these risks.  

We also identify techniques to slow down such attacks, inspired by strategies like password peppering and dedicated password hashing functions. These approaches include expanding the search space, effectively improving entropy, or designing obfuscation mechanisms and biometric cryptosystems that introduce artificial computational overhead, making the evaluation of a match more time-consuming.  


\vspace{-3mm}

\bibliographystyle{ACM-Reference-Format}
\bibliography{biblio.bib}

\appendix

\section{Recall on Classical Results in Statistics}
In this section, we provide a brief overview of the statistical results used in this paper.

\begin{figure*}
  \centering
  \includegraphics[width=\textwidth]{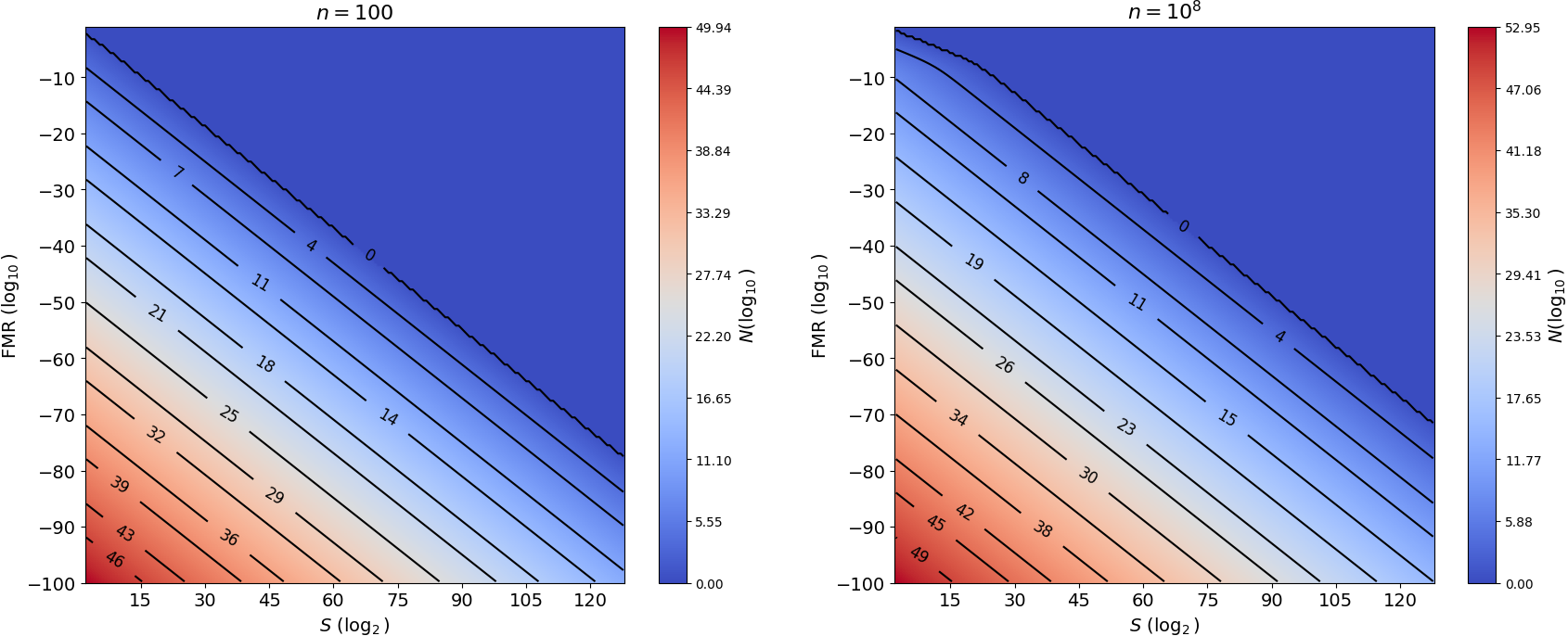}
  \caption{Critical population $N$ as a function of the system accuracy $\FMRe$ estimated on $n$ comparisons with a significance level of $95\%$ and the security level $S$ against \textit{untargeted attacks}.}
  \label{fig:critical_N1_untar}
\end{figure*}

\subsection{Central Limit Theorem (CLT)}
\label{sec:CLT:appendix}
The Central Limit Theorem~\cite{ross2017introductory} (CLT)
states that, under reasonable assumptions about the unknown distribution of the data, and provided the sample size is sufficiently large, the distribution of the sample mean converges almost surely to a normal distribution. 
In other words, the normal distribution acts as a universal attractor in this context.


\begin{theo}[Central Limit Theorem (CLT)]
\label{thm:CLT}\phantom{gotolinegotogot}
Let $(X_1, \dots, X_n)$ be a sequence of independent and identically distributed random variables with a finite expected value $\mu$ (\textit{i.e.}, $\mathbb{E}[X_i] = \mu$) and finite variance $\sigma^2$ (\textit{i.e.}, $\mathbb{V}[X_i] = \sigma^2 < \infty$). 
Then, as $n$ approaches infinity, the sample mean $$X = \frac{1}{n} \sum_{i=1}^n X_i$$ converges in distribution to a normal distribution $X \sim \mathcal{N}\left(\mu, \sigma^2/n\right)$, where $\mathcal{N}(\mu, \sigma^2/n)$ denotes the normal distribution with mean $\mu$ and variance $\sigma^2/n$.
\end{theo}

\subsection{Confidence Interval for a Normal Distribution with Unknown Mean and Variance}
\label{sec:CI:appendix}

A confidence interval~\cite{rees2018essential} is defined as an estimated range within which the true value of the parameter is expected to lie, based on a given confidence level. For example, if a $95\%$ confidence interval for the mean of a normally distributed population is computed, we can be $95\%$ confident that the true mean lies within this interval. Confidence intervals are most often computed for confidence levels of $95\%$ or $99\%$, but any significance level $\alpha$ can be chosen within the interval $]0,1[$.

\begin{theo}[Confidence Interval for the Mean of a Normal Distribution with Unknown Mean and Variance]
  \label{thm:CI}
  Let $X\sim\mathcal{N}(\mu,\sigma^2)$ for unknown $\mu$ and $\sigma$. Then, given an estimation $\overline{\mu}$ of $\mu$ and $\overline{\sigma}$ of $\sigma$ from $n$ samples, the confidence interval on $\overline{\mu}$ with risk $\alpha$ gives $$\mu \in \left[\overline{\mu}-\dfrac{c_\alpha\times\overline{\sigma}}{\sqrt{n}},\overline{\mu}+\dfrac{c_\alpha\times\overline{\sigma}}{\sqrt{n}}\right].$$
  $c_\alpha$ denotes the quantile of the Student's t-distribution $T$ with $n-1$ degrees of freedom (\textit{i.e.},$\mathbb{P}_T(-c_\alpha\leq T \leq c_\alpha)=1-\alpha$). This is equivalent to
  \begin{eqnarray*}
    \mathbb{P}\left(\overline{\mu}-\dfrac{c_\alpha\times\overline{\sigma}}{\sqrt{n}}\leq \mu \leq \overline{\mu}+\dfrac{c_\alpha\times\overline{\sigma}}{\sqrt{n}}\right)=1-\alpha.
  \end{eqnarray*}
  \end{theo}  

The values of $c_\alpha$ for various degrees of freedom ($df$) and significance levels ($\alpha$) are provided in Table~\ref{tab:t-distribution}.

\begin{table}
  \centering
  \resizebox{\columnwidth}{!}{
  \begin{tabular}{r|rrrrrrrr}
    \multicolumn{1}{c}{}         & \multicolumn{8}{c}{\textbf{Level of Significance}} \\
    \textbf{df}         & \textbf{0.25} & \textbf{0.20} & \textbf{0.15} & \textbf{0.10} & \textbf{0.05} & \textbf{0.025} & \textbf{0.01} & \textbf{0.005} \\ \cline{1-9} 
    \textbf{1}        & $1.000$ & $1.376$ & $1.963$ & $3.078$ & $6.314$ & $12.706$ & $31.821$ & $63.657$  \\
    \textbf{2}        & $0.816$ & $1.061$ & $1.386$ & $1.886$ & $2.920$ & $4.303$ & $6.965$ & $9.925$  \\
    \textbf{3}        & $0.765$ & $0.978$ & $1.250$ & $1.638$ & $2.353$ & $3.182$ & $4.541$ & $5.841$  \\
    \textbf{4}        & $0.741$ & $0.941$ & $1.190$ & $1.533$ & $2.132$ & $2.776$ & $3.747$ & $4.604$  \\
    \textbf{5}        & $0.727$ & $0.920$ & $1.156$ & $1.476$ & $2.015$ & $2.571$ & $3.365$ & $4.032$  \\
    \textbf{6}        & $0.718$ & $0.906$ & $1.134$ & $1.440$ & $1.943$ & $2.447$ & $3.143$ & $3.707$  \\
    \textbf{7}        & $0.711$ & $0.896$ & $1.119$ & $1.415$ & $1.895$ & $2.365$ & $2.998$ & $3.499$  \\
    \textbf{8}        & $0.706$ & $0.889$ & $1.108$ & $1.397$ & $1.860$ & $2.306$ & $2.896$ & $3.355$  \\
    \textbf{9}        & $0.703$ & $0.883$ & $1.100$ & $1.383$ & $1.833$ & $2.262$ & $2.821$ & $3.250$  \\
    \textbf{10}        & $0.700$ & $0.879$ & $1.093$ & $1.372$ & $1.812$ & $2.228$ & $2.764$ & $3.169$  \\
    \textbf{20}        & $0.687$ & $0.860$ & $1.064$ & $1.325$ & $1.725$ & $2.086$ & $2.528$ & $2.845$  \\
    \textbf{30}        & $0.683$ & $0.854$ & $1.055$ & $1.310$ & $1.697$ & $2.042$ & $2.457$ & $2.750$  \\
    \textbf{40}        & $0.681$ & $0.851$ & $1.050$ & $1.303$ & $1.684$ & $2.021$ & $2.423$ & $2.704$  \\
    \textbf{50}        & $0.679$ & $0.849$ & $1.047$ & $1.299$ & $1.676$ & $2.009$ & $2.403$ & $2.678$  \\
    \textbf{60}        & $0.679$ & $0.848$ & $1.045$ & $1.296$ & $1.671$ & $2.000$ & $2.390$ & $2.660$  \\
    \textbf{70}        & $0.678$ & $0.847$ & $1.044$ & $1.294$ & $1.667$ & $1.994$ & $2.381$ & $2.648$  \\
    \textbf{80}        & $0.678$ & $0.846$ & $1.043$ & $1.292$ & $1.664$ & $1.990$ & $2.374$ & $2.639$  \\
    \textbf{90}        & $0.677$ & $0.846$ & $1.042$ & $1.291$ & $1.662$ & $1.987$ & $2.368$ & $2.632$  \\
    \textbf{100}        & $0.677$ & $0.845$ & $1.042$ & $1.290$ & $1.660$ & $1.984$ & $2.364$ & $2.626$  \\
    \textbf{500}        & $0.675$ & $0.842$ & $1.038$ & $1.283$ & $1.648$ & $1.965$ & $2.334$ & $2.586$  \\
    \textbf{1000}        & $0.675$ & $0.842$ & $1.037$ & $1.282$ & $1.646$ & $1.962$ & $2.330$ & $2.581$  \\ 
    \textbf{5000}        & $0.675$ & $0.842$ & $1.037$ & $1.282$ & $1.645$ & $1.960$ & $2.327$ & $2.577$  \\
    \textbf{10000}        & $0.675$ & $0.842$ & $1.036$ & $1.282$ & $1.645$ & $1.960$ & $2.327$ & $2.576$  \\
    \textbf{$\infty$} & $0.674$         & $0.842$         & $1.036$        & $1.282$        & $1.645$        & $1.960$         & $2.326$        & $2.576$   \\\cline{1-9}
    & \multicolumn{1}{c}{\textbf{75\%}} & \multicolumn{1}{c}{\textbf{80\%}} & \multicolumn{1}{c}{\textbf{85\%}} & \multicolumn{1}{c}{\textbf{90\%}} & \multicolumn{1}{c}{\textbf{95\%}} & \multicolumn{1}{c}{\textbf{97.5\%}} & \multicolumn{1}{c}{\textbf{99\%}} & \multicolumn{1}{c}{\textbf{99.5\%}} \\
    \multicolumn{1}{c}{}          & \multicolumn{8}{c}{\textbf{Level of Confidence}} \\
  \end{tabular}}
  \caption{Quantiles $c_\alpha$ of the Student's $t$-distribution for various degrees of freedom ($df$) and significance levels ($\alpha$).}
  \label{tab:t-distribution}
\end{table}

\section{Numerical Analysis of Untargeted Attacks with Respect to the Confidence Interval}
\label{Annexe:ICimpact}

We analyze the \textit{critical population size}, \(N\), as a function of the \textit{False Match Rate} (\(\FMR\)), the confidence interval bounds derived from the number of comparisons \(n\), and the desired security level \(S\). Figure~\ref{fig:critical_N1_untar} illustrates how varying the number of comparisons influences the critical population size.
For example, systems performing \(n = 10^3\) comparisons achieve \textit{up to 10 bits of security} with a \(95\%\) confidence level and a critical population size below \(10^3\). In contrast, systems performing \(n = 10^8\) comparisons can achieve \textit{up to 20 bits of security} under the same confidence level and critical population size.
These results demonstrate that the confidence interval introduces significant variability in the critical population size a system can accommodate. Specifically, increasing the number of comparisons by a factor of \(10^5\) results in a discrepancy of approximately two orders of magnitude (a factor of \(100\)) in the critical population size.


\end{document}